\documentclass[12pt,reqno]{amsart}
\usepackage{amsmath}
\usepackage{latexsym}
\usepackage{amsfonts}
\usepackage{amssymb}
\usepackage{bbm,dsfont}

\theoremstyle{plain}
\newtheorem{proposition}{Proposition}
\newtheorem{theorem}{Theorem}
\newtheorem{lemma}{Lemma}
\newtheorem{corollary}{Corollary}
\theoremstyle{definition}

\newtheorem{remark}{Remark}
\newtheorem{example}{Example}
\newtheorem{definition}{Definition}

\newcommand{\pp}{\mathcal{P}}
\newcommand{\xx}{\mathcal{X}}
\newcommand{\yy}{\mathcal{Y}}

\newcommand{\cc}{\mathcal{C}}

\newcommand{\bb}{\mathcal{B}}
\newcommand{\ee}{\mathcal{E}}
\newcommand{\ff}{\mathcal{F}}
\newcommand{\hh}{\mathcal{H}} 
\newcommand{\kk}{\mathcal{K}}
\newcommand{\ti}{\mathcal{T}}

\newcommand{\mm}{\mathcal{M}} 
\newcommand{\hi}{\mathcal{H}} 
\newcommand{\ki}{\mathcal{K}} 
\newcommand{\lh}{\mathcal{L(H)}} 
\newcommand{\trh}{\mathcal{T(H)}} 
\newcommand{\ip}[2]{\left\langle \,#1\,\right|\left.\,#2\, \right\rangle} 

\newcommand{\kb}[2]{|#1\,\rangle\langle\,#2|} 
\newcommand{\no}[1]{\left\|#1\right\|} 
\newcommand{\notr}[1]{\no{#1}_{\textrm{tr}}} 
\newcommand{\tr}[1]{{\rm tr}\, (#1)} 
\newcommand{\de}{\,{\rm d}}
\newcommand{\ran}{\textrm{ran}\,} 

\newcommand{\M}{\mathsf{M}} 
\newcommand{\E}{\mathsf{E}} 
\renewcommand{\P}{\mathsf{P}} 

\newcommand{\f}{\mathcal{F}} 

\newcommand{\C}{\mathbb C} 
\newcommand{\N}{\mathbb N} 
\newcommand{\Z}{\mathbb Z} 
\newcommand{\half}{\frac{1}{2}}

\newcommand{\traccabis}[2]{\textrm{Tr}_{\hh_{#1}} \, (#2)}

\newcommand{\contrbis}[2]{{\rm ctr}_{\hh_{#1}} \, (#2)}

\newcommand{\elle}[1]{\mathcal{L}(#1)}

\newcommand{\boro}{\mathcal{B}(\Omega)}

\newcommand{\spanno}[1]{\textrm{span}\, \left\{ #1 \right\}}
\newcommand{\spannochiuso}[1]{\overline{\textrm{span}}\, \left\{ #1 \right\}}

\newcommand{\frecc}{\rightarrow}

\newcommand{\lft}{\left(}
\newcommand{\rgt}{\right)}
\newcommand{\scal}[2]{\left\langle\,#1\, ,\,#2\,\right\rangle}
\newcommand{\ev}{\textrm{ev}}

\begin{document}

\title[Extremal covariant positive operator valued measures]{Extremal covariant positive operator valued measures: the case of a compact symmetry group}

\author[Carmeli]{Claudio Carmeli}
\address{Claudio Carmeli, Dipartimento di Fisica, Universit\`a di Genova, and I.N.F.N., Sezione di Genova, Via
Dodecaneso 33, 16146 Genova, Italy}
\email{carmeli@ge.infn.it}

\author[Heinosaari]{Teiko Heinosaari}
\address{Teiko Heinosaari, Department of Physics, 
University of Turku, FIN-20014 Turku, Finland and Research Center for Quantum Information, Slovak Academy of Sciences, Bratislava, Slovakia}
\email{temihe@utu.fi}

\author[Pellonp\"a\"a]{Juha-Pekka Pellonp\"a\"a}
\address{Juha-Pekka Pellonp\"a\"a, Department of Physics, University of Turku, 
FIN-20014 Turku, Finland}
\email{juhpello@utu.fi}

\author[Toigo]{Alessandro Toigo}
\address{Alessandro Toigo, Dipartimento di Informatica, Universit\`a di Genova, and I.N.F.N., Sezione di Genova, Via
Dodecaneso 35, 16146 Genova, Italy}
\email{toigo@ge.infn.it}

\date{}

\maketitle

\begin{abstract}
Given a unitary representation $U$ of a compact group $G$ and a transitive $G$-space $\Omega$, we characterize the extremal elements of the convex set of all $U$-covariant positive operator valued measures.
\end{abstract}

\section{Introduction}\label{Intro}
 
In the modern theory of quantum mechanics, observables are represented as normalized positive operator valued measures (POVMs). The set of all POVMs having the same outcome space has natural convex structure. A convex mixture of two POVMs corresponds to a random choice between two measurement apparatuses. An extremal POVM thus describes an observable which is unaffected by this kind of randomness. 

In many applications one is interested in observables having some symmetry property. This is conventionally formulated as a covariance requirement with respect to a symmetry group. The covariance can arise from the symmetry of a particular problem \cite{PSAQT82} or the covariance can also be the defining property of some class of observables \cite{OQP97}. In this kind of situations the relevant set of observables is therefore the set of all covariant observables.

The determination of extremal covariant POVMs has long been a problem \cite{Holevo79}. In the case of a compact symmetry group $G$ the following results have been obtained earlier: When $G=\mathbb T$ acts on itself and the representation of $G$ has no multiplicities, the characterization of extremal covariant operator valued measures follows from \cite[Theorem~1]{LiTa94} if the representation space is finite-dimensional; in the infinite-dimensional case, the characterization is given in \cite[Theorem~1]{KiPe06}. If a compact group $G$ has a finite-dimensional representation, the determination of extremals is solved in \cite[Theorem~2]{Dariano04} in the case of $G$ acting on itself and in \cite[Theorem~3]{ChDa04} in the case of a transitive $G$-space. In this work we give a complete characterization of extremal covariant POVMs in the case of a compact group $G$ and an arbitrary representation (Theorem~\ref{th:notin} and Corollary~\ref{cor:ext}).

Our analysis of this problem proceeds in the following way. 
In Section~\ref{Basic} we fix the notations and recall some
concepts which are essential for our investigation. In Section~\ref{Structure} we give two characterizations of the structure of covariant POVMs. We remark that the characterization by means of families of sesquilinear forms (Section~\ref{Sec. Second characterization}) was already estabilished by Holevo in \cite{Kholevo87}. Finally, Section~\ref{Extremal} contains the main results which characterizes the extremal covariant POVMs. These are obtained by setting up a correspondence between the set of covariant POVMs and a family of reproducing kernel Hilbert spaces (RKHSs) of vector valued functions. In this way, a POVM is extremal if and only if its associated RKHS satisfies a particular property (Corollary~\ref{cor:ext}). In the finite dimensional case, such property is the algebraic condition found in \cite{Dariano04} and \cite{ChDa04}. We also show that a specific class of kernels, so-called rank 1 kernels, are always extremal.

\section{Basic definitions}\label{Basic}

For any complex Banach space $\xx$, we denote by $\no{\cdot}$ its norm and by $\xx^\ast$ its topological dual space. If $\yy$ is another Banach space, we let $\elle{\xx ; \yy}$ be the Banach space of bounded linear operators from $\xx$ to $\yy$, endowed with the uniform norm. If $\xx = \yy$, we use the abbreviated notation $\elle{\xx} \equiv \elle{\xx,\xx}$, and we denote by $O$ and $I$ the zero element and the identity of $\elle{\xx}$, respectively.
If $A\in\elle{\xx ; \yy}$, we write $A^\prime \in \elle{\yy^\ast ; \xx^\ast}$ for its adjoint.

If $\xx = \hh$ is a Hilbert space, we always assume that it is separable. We denote by $\ip{\cdot}{\cdot}$ the scalar product of $\hh$, and we take it linear in the second entry. When more than one space are involved and some confusion could arise, we add an index to norms and scalar products referring to the space under consideration. If $\kk$ is another Hilbert space and $A\in\elle{\hh;\kk}$, we denote by $A^\ast \in\elle{\kk ; \hh}$ the Hilbert space adjoint of $A$. For two selfadjoint operators $A,B\in\lh$, the relation $A\leq B$ means that $\ip{\psi}{A\psi} \leq \ip{\psi}{B\psi}$ for all $\psi\in\hi$. We denote by $\trh$ the Banach space of trace class operators on $\hh$ with the trace class norm $\notr{\cdot}$. An operator $T\in\trh$ is a \emph{state} if $T$ is positive and $\tr{T}=1$. 

\begin{definition}\label{def. di POVM 1}
Let $\Omega$ be a topological space and $\boro$ the Borel $\sigma$-algebra on $\Omega$. A mapping $\E:\boro\to\lh$ is a \emph{normalized positive operator valued measure (POVM)} on $\Omega$ if
\begin{itemize}
\item[(1)] $O\leq \E(X) \leq I$ for any $X\in\boro$;
\item[(2)] $\E(\Omega)=I$;
\item[(3)] $\E(\cup_i X_i)=\sum_i \E(X_i)$, the sum converging in the weak operator topology, for any sequence $(X_i)$ of disjoint Borel sets.
\end{itemize} 
\end{definition}
We remark that, by Ref.~\cite{DS}, p.~318, the sum in item (3) converges in the weak operator topology if and only if it converges in the strong operator topology.

We shortly recall the physical interpretation of a POVM as an observable. Loosely speaking, an observable is something which attaches a measurement outcome probability distribution to every state. If $\E$ is a POVM and $T$ a state, then the formula 
\begin{equation*}
p^\E_T(X)=\tr{T\E(X)},\quad X\in\boro,
\end{equation*}
defines a probability measure $p^\E_T$ on $\Omega$. Hence, $\E$ determines a mapping $T\mapsto p^\E_T$ from the set of states into the set of probability measures. This mapping is affine, i.e., it maps convex combinations of states to convex combinations of corresponding probability measures. Moreover, all affine mappings from the set of states into the set of probability measures can be represented as POVMs and this correspondence is one-to-one.

A special case of a POVM is a normalized projection valued measure (PVM). We recall that a POVM $\E$ is a PVM exactly when $\E(X\cap Y)=\E(X)\E(Y)$ for all $X,Y\in\boro$. PVMs are often referred as sharp observables.
 
In the following, $G$ is a fixed compact topological group, which is Hausdorff and satisfies the second axiom of countability. We also fix a closed subgroup $H$ of $G$ and, in the rest of this investigation, $\Omega$ is the quotient space $G/H$. The sets $H$ and $\Omega$ are compact, Hausdorff and second countable topological spaces in a natural way. We denote by $q$ the canonical projection from $G$ onto $\Omega$, and we also write $\dot{g}\equiv q(g)$.
The normalized Haar measures of $G$ and $H$ are denoted by $\mu_G$ and $\mu_H$, respectively.  The normalized $G$-invariant measure on $\Omega$ is denoted by $\mu_\Omega$. The following relation (see e.g. \cite{CAHA95}) will be used later:
\begin{equation}\label{Mackey-Bruhat}
\int_G f(g) \de \mu_G (g) = \int_\Omega \int_H f(gh) \de \mu_H (h) \de \mu_\Omega (\dot{g})
\end{equation}
for all $f\in L^1 ( G , \mu_G )$.

\begin{definition}
Let $U$ be a (strongly continuous) unitary representation of $G$ in a Hilbert space $\hi$. A POVM $\E:\boro \frecc \lh$ is {\em covariant} with respect to $U$ (or {\em $U$-covariant}, for short) if 
\begin{equation}
U(g) \E(X) U(g)^\ast = \E(gX)
\end{equation}
for all $g\in G$ and $X\in\boro$.
\end{definition}

As an example of $U$-covariant POVMs we recall the following result of Davies \cite[Section 4.5]{QTOS76}. Assume that the dimension of the Hilbert space $\hi$ is finite. Then the non normalized (i.e. not necessarily satisfying $\E(\Omega)=I$) $U$-covariant POVMs are in one-to-one correspondence with the positive operators $C\in\lh$ such that $CU(h)=U(h)C$ for all $h\in H$. The correspondence is given by the formula 
\begin{equation}\label{UCU}
\E(X)=\int_{q^{-1}(X)} U(g)CU(g)^\ast\ d\mu_G(g).
\end{equation}
The normalization condition now reads
\begin{equation}
\int_{G} U(g)CU(g)^\ast\ d\mu_G(g)=I.
\end{equation}

If $\dim\hi=\infty$, the formula (\ref{UCU}) still defines a $U$-covariant POVM. However, this is not an exhaustive characterization anymore.

\section{Structure of covariant POVMs}\label{Structure}

\subsection{Preliminaries and notations}

We denote by $\hat{G}$ the unitary dual of $G$, i.e. the (denumerable) set of (equivalence classes of) irreducible unitary representations of $G$. For each $\pi\in \hat{G}$, we let $\hi_\pi$ be the representation space of $\pi$ and set $d_\pi = \dim \hi_\pi<\infty$.

Let $U$ be a unitary representation of $G$ acting in the Hilbert space $\hi$. Then there exists a sequence of Hilbert spaces $( \ki_\pi )_{\pi\in\hat{G}}$ such that
$$
\hi = \oplus_{\pi\in\hat{G}} \hi_\pi \otimes \ki_\pi\quad \text{and}\quad U = \oplus_{\pi\in\hat{G}} \pi \otimes I_{\ki_\pi} .
$$
The dimension of $\ki_\pi$ is the {\em multiplicity} of $\pi$ in $U$, which is uniquely determined.

For convenience, if $A_\pi \in\elle{\hh_\pi \otimes \kk_\pi}$, we denote by $A_\pi$ also the bounded operator on $\hh$ that is $0$ on $( \hh_\pi \otimes \kk_\pi )^\perp$ and equals $A_\pi$ on $\hh_\pi \otimes \kk_\pi$.

Fix $\pi\in\hat{G}$. Since $\dim \hh_\pi < \infty$, we have the identifications $\hi_\pi\otimes\hi_\pi^\ast = \elle{\hi_\pi} = \ti (\hh_\pi)$, and the trace ${\rm tr}_{\hi_\pi} : \hi_\pi\otimes\hi_\pi^\ast \frecc \C$ is bounded. Let $\ki$ be a Hilbert space. The mapping
$$
E_\pi : \kk \frecc \oplus_{\rho\in\hat{G}} \hi_\rho\otimes\hi_\rho^\ast \otimes \kk ,\quad k\mapsto E_\pi k = I_{\hh_\pi} \otimes k
$$
is thus bounded. We denote ${\rm ctr}_{\hi_\pi} := E_\pi^\ast : \oplus_{\rho}\  \hi_\rho\otimes\hi_\rho^\ast \otimes \kk \frecc \kk$, so ${\rm ctr}_{\hi_\pi}$ is the contraction with respect to $\hi_{\pi}$. Clearly,
\begin{equation*}
{\rm ctr}_{\hi_\pi} (\phi_\rho) = \delta_{\rho \pi} ({\rm tr}_{\hi_\rho} \otimes I_\kk ) \phi_\rho \quad \forall \phi_\rho\in\hi_\rho\otimes\hi_\rho^\ast \otimes \kk .
\end{equation*}

We define also bounded linear maps 
$$
\ee_\pi : \ti (\kk_\pi) \frecc \ti (\hh), \qquad \ee_\pi T_\pi = I_{\hh_\pi} \otimes T_\pi,
$$
and $\textrm{Tr}_{\hh_\pi} = \ee_\pi^\prime : \elle{\hh} \frecc \elle{\kk_\pi}$. Explicitly, if $A_\rho\in\elle{\hh_\rho}$ and $B_\rho\in\elle{\kk_\rho}$, then
$$
\traccabis{\pi}{A_\rho \otimes B_\rho} = \delta_{\rho \pi} ( {\rm tr}_{\hh_\rho} A_\rho ) B_\rho .
$$

\subsection{First characterization}

\begin{theorem}\label{La proposizione}
Fix an infinite dimensional Hilbert space $\kk$. Let $\E:\boro \frecc \lh$ be a $U$-covariant POVM. Then there exists a family of isometries $V_\pi : \kk_\pi \frecc \hh_\pi^\ast \otimes \kk$ labeled by $\pi\in \hat{G}$ such that
\begin{equation}
\ip{w_\rho}{\E(X) v_\pi}
= \sqrt{d_\pi d_\rho} \int_{q^{-1} (X)} \ip{\contrbis{\rho}{(\rho(g)^{-1} \otimes V_\rho) w_\rho}}{ \contrbis{\pi}{(\pi(g)^{-1} \otimes V_\pi) v_\pi}} \de \mu_G (g) \label{La POVM}
\end{equation}
for all $X\in\boro$ and for all $v_\pi\in \hh_\pi \otimes \kk_\pi$, $w_\rho\in \hh_\rho\otimes\kk_\rho$.

Conversely, if $V_\pi : \kk_\pi \frecc \hh_\pi^\ast \otimes \kk$ is a family of isometries labeled by $\pi\in \hat{G}$, then equation~(\ref{La POVM}) defines a $U$-covariant POVM on $\Omega$.
\end{theorem}

Before proving Theorem~\ref{La proposizione}, we recall some notions from the theory of induced representations. Let $\lambda$ be the left regular representation of $G$ in $L^2 (G,\mu_G)$, and $\P:\boro \frecc \elle{L^2 (G,\mu_G)}$ be the following $\lambda$-covariant PVM:
\begin{equation*}
[\P(X) f] (g) = \chi_{q^{-1}(X)} (g) f(g), \quad f\in L^2 (G,\mu_G).
\end{equation*}
(Here $\chi_{q^{-1}(X)}$ is the characteristic function of the set $q^{-1}(X)$ in $G$).
If $\sigma$ is a unitary representation of $H$ acting in the Hilbert space $\kk_0$, the {\em representation of $G$ induced by $\sigma$} is the left regular representation $\lambda\otimes I_{\kk_0}$ of $G$ in $L^2 (G,\mu_G) \otimes \kk_0 = L^2 (G,\mu_G;\kk_0)$ restricted to the closed invariant subspace
\begin{equation*}
\hh^\sigma = \left\{ f\in L^2 (G,\mu_G;\kk_0) \mid \forall h\in H,\, f(gh) = \sigma(h)^{-1} f(g)\textrm{ for almost any }g\in G \right\}.
\end{equation*}
We denote by $\lambda^\sigma$ such restriction. Clearly $\hh^\sigma$ is invariant for the PVM $\P \otimes I_{\kk_0}$, and the restriction $\P^\sigma$ of $\P \otimes I_{\kk_0}$ to $\hh^\sigma$ is the {\em canonical PVM induced by $\sigma$}. It is easy to see that $\P^\sigma$ is $\lambda^\sigma$-covariant.

\begin{proof}[Proof of Theorem~\ref{La proposizione}]
Let $\E$ be a $U$-covariant POVM on $\Omega$. By the imprimitivity theorems of Mackey \cite{Mackey52} and Cattaneo \cite{Cattaneo79}, there exists a unitary representation $\sigma$ of the subgroup $H$ acting in a Hilbert space $\kk_0$, and an isometry $W:\hh\frecc\hh^\sigma$ intertwining $U$ with $\lambda^\sigma$ such that
\begin{equation*}
\E(X) = W^\ast \P^\sigma (X) W \quad \forall X\in \boro.
\end{equation*}

Let $\imath:\kk_0 \frecc \kk$ be an injection of Hilbert spaces ($\imath$ exists since $\dim\kk = \infty$ is assumed).
Composing $W$ with the injections $\hh^\sigma \hookrightarrow L^2(G,\mu_G) \otimes \kk_0$ and $L^2(G,\mu_G) \otimes \kk_0 \stackrel{I\otimes \imath}{\hookrightarrow} L^2(G,\mu_G) \otimes \kk$, we find an isometry $\tilde{W} : \hh \frecc L^2 (G,\mu_G;\kk)$ intertwining $U$ with $\lambda\otimes I_\kk$ such that
\begin{equation}\label{eq. di M. C.}
\E(X) = \tilde{W}^\ast ( \P (X) \otimes I_\kk ) \tilde{W} \quad \forall X\in \boro.
\end{equation}
Conversely, if $\tilde{W} : \hh \frecc L^2 (G,\mu_G;\kk)$ intertwines $U$ with $\lambda\otimes I_\kk$, then  equation~(\ref{eq. di M. C.}) defines a $U$-covariant POVM $\E$ on $\Omega$.
The problem of characterising all $U$-covariant POVMs on $\Omega$ is thus seen to be equivalent to the problem of finding all intertwining isometries between $U$ and $\lambda\otimes I_\kk$.

By Fourier-Plancherel theory, there is a unitary isomorphism
\begin{equation*}
\overline{\f}^\ast : \oplus_{\pi\in\hat{G}} \hh_\pi \otimes \hh_\pi^\ast \otimes \kk \frecc
L^2 (G,\mu_G;\kk)
\end{equation*}
intertwining the representation $\oplus_{\pi\in\hat{G}} \pi \otimes I_{\hh_\pi^\ast} \otimes I_\kk$ with $\lambda \otimes I_\kk$, whose action on 
$\phi_\pi \in \hh_\pi \otimes \hh_\pi^\ast \otimes \kk$ is
\begin{equation*}
\overline{\f}^\ast \phi_\pi (g) = \sqrt{d_\pi}\, \contrbis{\pi}{(\pi (g)^{-1} \otimes I_{\hh_\pi^\ast} \otimes I_\kk ) \phi_\pi}
\end{equation*}
($\overline{\f}^\ast$ is the inverse Fourier-Plancherel cotransform of $L^2 (G,\mu_G)$ tensored with the identity of $\kk$).
For each $\pi\in\hat{G}$, let $V_\pi : \kk_\pi \frecc \hh_\pi^\ast \otimes \kk$ be an isometry. Then $V = \oplus_{\pi} I_{\hh_\pi} \otimes V_\pi$ is an isometry from $\hh$ into $\oplus_{\pi} \hh_\pi \otimes \hh_\pi^\ast \otimes \kk$ which intertwines the representations $U$ and $\oplus_{\pi} \pi \otimes I_{\hh_\pi^\ast} \otimes I_\kk$. Moreover, every isometry intertwining $U$ and $\oplus_{\pi} \pi \otimes I_{\hh_\pi^\ast} \otimes I_\kk$ has this form for some choice of the isometries $V_\pi$.

Fixed the sequence of isometries $\{V_\pi\}_{\pi\in\hat{G}}$, the corresponding intertwining isometry $\tilde{W} : \hh \frecc L^2 (G,\mu_G;\kk)$ is thus
\begin{equation*}
\tilde{W}v_\pi = \overline{\f}^\ast V v_\pi = \sqrt{d_\pi}\, \contrbis{\pi}{( \pi (\,\cdot\,)^{-1} \otimes V_\pi ) v_\pi }
\end{equation*}
for $v_\pi \in \hh_\pi \otimes \kk_\pi$.

If $v_\pi \in \hh_\pi \otimes \kk_\pi$ and $w_\rho \in \hh_\rho \otimes \kk_\rho$, equation~(\ref{eq. di M. C.}) then gives
\begin{eqnarray*}
\ip{w_\rho}{\E(X) v_\pi} & = & 
\ip{\tilde{W} w_\rho}{( \P(X) \otimes I_\kk ) \tilde{W} v_\pi} \\
& = & \sqrt{d_\pi d_\rho} \int_{q^{-1} (X)} \ip{\contrbis{\rho}{( \rho(g)^{-1} \otimes V_\rho ) w_\rho}}{\contrbis{\pi}{( \pi(g)^{-1} \otimes V_\pi ) v_\pi}} \de \mu_G (g)
\end{eqnarray*}
as claimed.
\end{proof}

To illustrate the content of Theorem~\ref{La proposizione}, we take a closer look on two special cases. We note that in these cases of an abelian group, the characterization of covariant POVMs is also given in \cite{CaDeTo04}.

\begin{example}\label{ex:abelian}
Suppose that $G$ is a compact abelian group. Then also $\Omega$ is a compact abelian group in the natural way. For each $\pi\in \hat{G}$, the representation space $\hi_\pi$ is one dimensional and $\hat{G}$ is a discrete abelian group, i.e.~the character group of $G$. We denote by $H^\perp$ the annihilator of the subgroup $H$, that is,
\begin{equation*}
H^\perp = \left\{ \pi \in \hat{G} \mid \pi(h) = 1 \ \forall h\in H \right\}.
\end{equation*}
The annihilator $H^\perp$ is a subgroup of $\hat{G}$ and it can be identified with the character group $\hat{\Omega}$, the identification being $\pi (q(g)) : = \pi (g)$ for all $g\in G$ and $\pi \in H^\perp$. Finally, let $\f_{\Omega}$ be the Fourier transform of $\Omega$, i.e.~for $f\in L^1 (\Omega,\mu_{\Omega})$,
\begin{equation*}
\left( \f_{\Omega} f \right) (\pi) = \int_{\Omega} f(\omega) \pi(\omega^{-1}) \, d\mu_{\Omega} (\omega) \quad \forall \pi \in 
H^\perp .
\end{equation*}

We have the identifications $\hh_\pi \otimes \kk_\pi = \kk_\pi$ and $\hh_\pi \otimes \hh_\pi^\ast \otimes \kk = \kk$. With the last identification we have ${\rm ctr}_{\hh_\pi} = I_\kk$.

By Theorem~\ref{La proposizione}, if $\E$ is a $U$-covariant POVM on $\Omega$, there exists a family of isometries $V_\pi : \kk_\pi \frecc \kk$ labeled by $\pi \in \hat{G}$, such that
\begin{eqnarray*}
\ip{w_\rho}{\E(X) v_\pi} &=& \int_{q^{-1} (X)} \left( \pi^{-1}\rho \right)(g) \ip{V_\rho w_\rho}{V_\pi w_\pi} \de\mu_G (g) \\
\textrm{\emph{by eq.}~(\ref{Mackey-Bruhat})} &=& \ip{V_\rho w_\rho}{V_\pi v_\pi} \int_X \int_H \left( \pi^{-1}\rho \right)(gh) \de\mu_H (h) \de\mu_{\Omega} (\dot{g}) \\
&=& \ip{V_\rho w_\rho}{V_\pi v_\pi} \delta_{\rho H^\perp}(\pi H^\perp)\int_X \left( \pi^{-1}\rho \right)(g) \de\mu_{\Omega} (\dot{g}) \\
&=& \ip{V_\rho w_\rho}{V_\pi v_\pi} \delta_{\rho H^\perp}(\pi H^\perp) \left( \f_{\Omega} \chi_X \right) \left( \rho^{-1}\pi \right)
\end{eqnarray*}
for $v_\pi \in \kk_\pi$, $w_\rho \in \kk_\rho$. (In the above equation, $\delta_{\rho H^\perp}(\pi H^\perp) = 1$ if $\rho$ and $\pi$ are in the same $H^\perp$-coset, and is $0$ otherwise; note that $\rho^{-1}\pi$ is in $H^\perp$ exactly when $\delta_{\rho H^\perp}(\pi H^\perp) = 1$).
\end{example}

\begin{example}\label{ex:multipfree}
Consider the situation of Example \ref{ex:abelian} in the case when there is a set $\Sigma \subseteq \hat{G}$ such that $\dim \kk_\pi = 1$ for all $\pi\in\Sigma$ and $\dim \kk_\pi = 0$ otherwise. For each $\pi\in\Sigma$, choose a unit vector $e_\pi \in \kk_\pi$. Since $V_\pi$ is an isometry, there is a unit vector $v_\pi \in \kk$ such that $V_\pi = \kb{v_\pi}{e_\pi}$, and we have
\begin{equation*}
\ip{e_\rho}{\E(X) e_\pi} = \ip{v_\rho}{v_\pi} \delta_{\rho H^\perp}(\pi H^\perp) \left( \f_{\Omega} \chi_X \right) \left( \rho^{-1}\pi \right).
\end{equation*}
In particular, if $H=\{ e \}$ the last equation is
\begin{equation*}
\ip{e_\rho}{\E(X) e_\pi} = \ip{v_\rho}{v_\pi} \left( \f_{\Omega} \chi_X \right) \left( \rho^{-1}\pi \right).
\end{equation*}
\end{example}

\subsection{Second characterization}\label{Sec. Second characterization}

We proceed to give an alternative characterization of covariant POVMs. Suppose $E$ and $\{V_\pi\}_{\pi\in\hat{G}}$ are as in Theorem~\ref{La proposizione}. Let $V = \oplus_{\pi} I_{\hh_\pi} \otimes V_\pi$.
For $\pi,\rho\in\hat{G}$, we introduce the following sesquilinear form $\Pi[\rho,\pi]$ on $\hh$:
\begin{equation*}
\Pi[\rho,\pi] \lft w , v \rgt = \sqrt{d_\pi d_\rho} \int_H \ip{\contrbis{\rho}{ VU(h) w}}{\contrbis{\pi}{ VU(h) v}} \de \mu_H (h).
\end{equation*}
Clearly, for all $v_\pi \in \hh_\pi\otimes\kk_\pi$, $w_\rho \in \hh_\rho\otimes\kk_\rho$,
\begin{equation*}
\Pi[\rho,\pi] \lft w_\rho , v_\pi \rgt = \sqrt{d_\pi d_\rho} \int_H \ip{\contrbis{\rho}{ (\rho (h) \otimes V_\rho ) w_\rho}}{\contrbis{\pi}{ (\pi (h) \otimes V_\pi) v_\pi}} \de \mu_H (h).
\end{equation*}
If $v_\sigma \in \hh_\sigma \otimes \kk_\sigma$, $w_{\sigma^\prime} \in \hh_{\sigma^\prime} \otimes \kk_{\sigma^\prime}$, and $\sigma \neq \pi$ or $\sigma^\prime \neq \rho$, then
$\Pi[\rho,\pi] (w_{\sigma^\prime} , v_\sigma) = 0$.

Equation~(\ref{La POVM}) can be rewritten as
\begin{eqnarray*}
\ip{w_\rho}{\E(X) v_\pi} & = & \sqrt{d_\pi d_\rho} \int_X \de \mu_\Omega (\dot{g}) \int_H  \left\langle\contrbis{\rho}{ ( \rho(h^{-1}) \rho(g^{-1}) \otimes V_\rho) w_\rho}\right| \\
&&\left|\contrbis{\pi}{ ( \pi(h^{-1}) \pi(g^{-1}) \otimes V_\pi) v_\pi}\right\rangle \de \mu_H (h) \\
& = & \sqrt{d_\pi d_\rho} \int_X \de \mu_\Omega (\dot{g}) \int_H  \left\langle\contrbis{\rho}{ (\rho(h^{-1}) \otimes V_\rho ) (\rho(g^{-1}) \otimes I_{\kk_\rho} ) w_\rho}\right| \\
&&\left|\contrbis{\pi}{ (\pi(h^{-1}) \otimes V_\pi ) (\pi(g^{-1}) \otimes I_{\kk_\pi} ) v_\pi}\right\rangle \de \mu_H (h) \\
& = & \int_X \Pi[\rho,\pi] \lft U(g)^{-1} w_\rho , U(g)^{-1} v_\pi \rgt \de\mu_\Omega (\dot{g}),
\end{eqnarray*}
where we have used the measure decomposition (\ref{Mackey-Bruhat}) and invariance of $\mu_H$.

The above defined sesquilinear form $\Pi[\rho,\pi]$ has the following properties:
\begin{enumerate}
\item[{\rm (i)}] $\Pi[\rho,\pi]$ is a bounded sesquilinear form on $\hh$.
\item[{\rm (ii)}] By invariance of $\mu_H$, we have
\begin{equation*}
\Pi[\rho,\pi] \lft U(h) w , U(h) v \rgt = \Pi[\rho,\pi] \lft w , v \rgt \quad \forall h\in H
\end{equation*}
for all $v,w\in\hh$.
\item[{\rm (iii)}] If $\{ v^\pi \}_{\pi\in\hat{G}}$ is a sequence of vectors in $\hh$ such that $v^\pi \neq 0$ only for a finite number of $\pi$'s, then
\begin{equation*}
\sum\nolimits_{\rho,\pi} \Pi[\rho,\pi] ( v^\rho , v^\pi ) \geq 0.
\end{equation*}
\item[{\rm (iv)}] For $v_\rho , w_\rho \in \hh_\rho \otimes \kk_\rho$
\begin{eqnarray}
&&\int_G \Pi[\pi,\pi] \lft U(g) w_\rho , U(g) v_\rho \rgt \de\mu_G (g) \notag \\
&&\qquad = \delta_{\pi \rho} \int_\Omega \Pi[\rho,\rho] \lft U(g)^{-1} w_\rho , U(g)^{-1} v_\rho \rgt \de\mu_\Omega (\dot{g}) \notag \\
&&\qquad = \delta_{\pi \rho} \ip{w_\rho}{\E (\Omega) v_\rho} \notag \\
&&\qquad = \delta_{\pi \rho} \ip{w_\rho}{v_\rho} \label{normal. di Pi} .
\end{eqnarray}
\end{enumerate}

\begin{remark}
If $\{ \Pi[\pi,\rho] \}_{\pi,\rho \in\hat{G}}$ is a family of sesquilinear forms on $\hh$ satisfyng conditions {\rm (i)}, {\rm (iii)} and {\rm (iv)} above, then
\begin{equation}
\Pi[\rho,\pi] (w_{\sigma^\prime} , v_\sigma) = 0 \quad \textrm{if } v_\sigma \in \hh_\sigma \otimes \kk_\sigma ,\, w_{\sigma^\prime} \in \hh_{\sigma^\prime} \otimes \kk_{\sigma^\prime} , \textrm{ and } \sigma\neq\pi \textrm{ or } \sigma^\prime \neq \rho. \label{annullamento di Pi}
\end{equation}
In fact, by {\rm (iii)} and Cauchy-Schwarz inequality
$$
| \Pi[\rho,\pi] (w_{\sigma^\prime} , v_\sigma) | \leq \Pi[\rho,\rho] (w_{\sigma^\prime} , w_{\sigma^\prime})^{1/2} \Pi[\pi,\pi] (v_\sigma , v_\sigma)^{1/2}.
$$
If $\sigma \neq \pi$ then, by eq.~(\ref{normal. di Pi}),
$$
\int_G \Pi[\pi,\pi] (U(g) v_\sigma , U(g) v_\sigma) \de\mu_G (g) = 0 ,
$$
which implies $\Pi[\pi,\pi] ( v_\sigma , v_\sigma) = 0$ since $\Pi[\pi,\pi]$ is positive semidefinite and the integrand is a continuous function. A similar discussion holds for the other factor.
\end{remark}

\begin{proposition}\label{La proposizione col kernel}
Suppose $\{ \Pi[\rho,\pi] \}_{\rho,\pi \in\hat{G}}$ is a set of sesquilinear forms on $\hh$ satisfyng conditions {\rm (i)}-{\rm (iv)} above. Then there is a $U$-covariant POVM $\E$ on $\Omega$ such that
\begin{equation}
\ip{w_\rho}{\E(X) v_\pi} = \int_X \Pi[\rho,\pi] \lft U(g)^{-1} w_\rho , U(g)^{-1} v_\pi \rgt \de\mu_\Omega (\dot{g})
\label{La POVM col kernel}
\end{equation}
for $v_\pi \in \hh_\pi\otimes\kk_\pi$, $w_\rho \in \hh_\rho\otimes\kk_\rho$.

Conversely, if $\E$ is a $U$-covariant POVM on $\Omega$, then there exist sesquilinear forms $\{ \Pi[\rho,\pi] \}_{\rho,\pi \in\hat{G}}$ on $\hh$ satisfyng conditions {\rm (i)}-{\rm (iv)} and such that eq.~(\ref{La POVM col kernel}) holds. The family of sesquilinear forms $\{ \Pi[\rho,\pi] \}_{\rho,\pi \in\hat{G}}$ is uniquely determined by $\E$.
\end{proposition}
\begin{proof}
For uniqueness, we see from eq.~(\ref{La POVM col kernel}) that $m_{w_\rho , v_\pi} (\dot{g}) := \Pi[\rho,\pi] \lft U(g) w_\rho , U(g) v_\pi \rgt$ is the density of the complex measure $X\mapsto \ip{w_\rho}{\E(X) v_\pi}$ with respect to $\mu_\Omega$. Since $m_{w_\rho , v_\pi}$ is a continuous function and the support of $\mu_\Omega$ is the whole $\Omega$, $m_{w_\rho , v_\pi}$ is unique. Hence, $\Pi[\rho,\pi](\cdot , \cdot) = m_{\cdot , \cdot} (e)$ is unique.

It remains to check that, given $\{ \Pi[\pi,\rho] \}_{\pi,\rho \in\hat{G}}$ as above, eq.~(\ref{La POVM col kernel}) defines a POVM $\E$. Let $\hh_0 = \spanno{v_\pi \mid v_\pi \in \hh_\pi \otimes \kk_\pi ,\, \pi\in\hat{G}}$. For $X\in\boro$, we introduce in $\hh_0$ the sesquilinear form
$$
\ip{w}{v}_X := \sum_{\pi,\rho} \int_X \Pi[\rho,\pi] (U(g)^{-1} w , U(g)^{-1} v ) \de\mu_\Omega (\dot{g}).
$$
(By eq.~(\ref{annullamento di Pi}), the sum involves only a finite number of $\pi,\rho\in\hat{G}$).
By {\rm (iii)}, $\ip{\cdot}{\cdot}_X$ is positive semidefinite. If $\{ X_i \}_{i\in\N}$ is a sequence of disjoint sets, then $\ip{w}{v}_{\cup_i X_i} = \sum_i \ip{w}{v}_{X_i}$, the sum converging absolutely. Finally, for $v = \sum_{\pi\in J} v_\pi$, with $v_\pi \in \hh_\pi \otimes \kk_\pi$ and $J\subseteq \hat{G}$ finite,
\begin{eqnarray*}
\ip{v}{v}_\Omega & = & \sum_{\pi,\rho \in J} \int_G \Pi[\rho,\pi] (U(g)^{-1} v_\rho , U(g)^{-1} v_\pi ) \de\mu_G (g) \\
\textrm{\emph{by Schur lemma}} & = & \sum_{\pi\in J} \int_G \Pi[\pi,\pi] (U(g)^{-1} v_\pi , U(g)^{-1} v_\pi ) \de\mu_G (g) \\
\textrm{\emph{by} {\rm (iv)}} & = & \sum_{\pi\in J} \ip{v_\pi}{v_\pi} \\
& = & \no{v_\pi}^2,
\end{eqnarray*}
i.e.\ $\ip{\cdot}{\cdot}_\Omega = \ip{\cdot}{\cdot}$ on $\hh_0$. Since $\ip{v}{v}_X \leq \ip{v}{v}_\Omega$, $\ip{\cdot}{\cdot}_X$ extends to a bounded sesquilinear form on $\hh$. Thus, $\ip{\cdot}{\cdot}_X = \ip{\E(X) \cdot}{\cdot}$ for some $\E(X) \in \elle{\hh}$. Positivity, $\sigma$-additivity and normalisation of $\E$ follow from the similar properties of the family of sesquilinear forms $\{ \ip{\cdot}{\cdot}_{X} \}_{X\in\boro}$.
\end{proof}

From Theorem~\ref{La proposizione} and Proposition~\ref{La proposizione col kernel} we obtain the following.
\begin{corollary}\label{forma di Pi}
If $V_\pi : \kk_\pi \frecc \hh_\pi^\ast \otimes \kk$ ($\pi\in\hat{G}$) is a family of isometries, and $V = \oplus_\pi I_{\hh_\pi} \otimes V_\pi$, then
$$
\Pi[\rho,\pi] \lft w , v \rgt = \sqrt{d_\pi d_\rho} \int_H \ip{\contrbis{\rho}{ VU(h) w}}{\contrbis{\pi}{ VU(h) v}} \de \mu_H (h)
$$
defines a family of sesquilinear forms $\{ \Pi[\rho,\pi] \}_{\rho,\pi \in\hat{G}}$ on $\hh$ satisfyng conditions {\rm (i)}-{\rm (iv)}, and every such family arises in this manner for some choice of the isometries $\{ V_\pi \}_{\pi \in\hat{G}}$.
\end{corollary}

\begin{remark}
The characterisation of the most general $U$-covariant POVM given in Proposition~\ref{La proposizione col kernel} is already contained in \cite{Kholevo87}. However, it can happen that the sesquilinear forms $\Pi [\rho , \pi]$ are not bounded uniformly in $\rho,\pi \in \hat{G}$, thus contradicting a part of the statement in Theorem 1 of \cite{Kholevo87} as shown in the following example.

Suppose $G = \rm SU(2)$, so that $\{d_\pi\}_{\pi\in\hat{G}} = \Z_+$ is an unbounded set. Let $H = \{1\}$, and fix $U$ such that $\kk_\pi = \hh_\pi^\ast$ for all $\pi\in\hat{G}$, so that $\hh = \oplus_\pi \hh_\pi \otimes \hh_\pi^\ast$. Choose a vector $k\in\kk$ with $\no{k} = 1$, and for all $\pi$ let $V_\pi : \hh_\pi^\ast \frecc \hh_\pi^\ast \otimes \kk$ be the following isometry
$$
V_\pi h^\ast_\pi = h^\ast_\pi \otimes k \quad \forall h^\ast_\pi \in \hh_\pi^\ast.
$$
Let $v_\pi = d_\pi^{-1/2} I_{\hh_\pi} \in \hh_\pi \otimes \hh_\pi^\ast$. Then $\no{v_\pi}_\hh = 1$, and
$\contrbis{\pi}{ (I_{\hh_\pi} \otimes V_\pi) v_\pi } = d_\pi^{1/2} k$.
It follows that
\begin{equation*}
\Pi[\rho,\pi] \lft v_\rho , v_\pi \rgt = \sqrt{d_\pi d_\rho} \ip{\contrbis{\rho}{ (I_{\hh_\rho} \otimes V_\rho) v_\rho }}{\contrbis{\pi}{ (I_{\hh_\pi} \otimes V_\pi) v_\pi }}
= d_\pi d_\rho ,
\end{equation*}
thus showing that $\Pi$ is not bounded uniformly in $\rho,\pi\in\hat{G}$. Moreover, if, following the notations of \cite{Kholevo87}, we introduce the linear subspace
$$
\hh^1 = \left\{ \sum_{\pi\in J} v_\pi \mid v_\pi \in \hh_\pi \otimes \kk_\pi ,\,  \sum_{\pi\in J} \no{v_\pi}^p < \infty \textrm{ for } p=1,2 \right\} \subset \hh ,
$$
it is easy to check that, for such $\Pi$, also the corresponding sesquilinear form $\Pi : \hh^1 \times \hh^1 \frecc \C$ is not defined.
\end{remark}

We define $K(\rho,\pi) : \hh \frecc \hh$ by the relation
$$
\ip{w}{K(\rho,\pi) v} = \Pi[\rho,\pi] (w,v) \quad \forall w,v \in\hh .
$$
The family of sesquilinear forms $\{ \Pi[\rho,\pi] \}_{\rho,\pi\in\hat{G}} \}$ is thus in one-to-one correspondence with the family of operators $\{ K(\rho,\pi) \}_{\rho,\pi\in\hat{G}} \}$ on $\hh$ with the following properties:

\begin{enumerate}
\item[{\rm (i')}] $K(\rho,\pi)$ are bounded operators on $\hh$.
\item[{\rm (ii')}] For all $h\in H$
\begin{equation*}
K(\rho,\pi) U(h) = U(h) K(\rho,\pi) .
\end{equation*}
\item[{\rm (iii')}] $K : \hat{G} \times \hat{G} \frecc \elle{\hh}$ is a kernel of positive type, i.e., if $\{ v^\pi \}_{\pi\in\hat{G}}$ is a sequence of vectors in $\hh$ such that $v^\pi \neq 0$ only for a finite number of $\pi$'s, then
\begin{equation*}
\sum\nolimits_{\rho,\pi}
\ip{v^\rho}{K(\rho,\pi) v^\pi}
\geq 0.
\end{equation*}
\item[{\rm (iv')}]
\begin{equation}\label{traccia di K}
\traccabis{\rho}{K(\pi,\pi)} = \delta_{\rho \pi} d_\rho I_{\kk_\rho}.
\end{equation}
In fact, by eq.~(\ref{annullamento di Pi}), $(\ker K (\pi , \pi) )^\perp = \ran K (\pi , \pi) \subseteq \hh_\pi\otimes\kk_\pi$. In particular, $\traccabis{\rho}{K(\pi,\pi)} = 0$ if $\pi\neq\rho$. By Schur lemma
\begin{equation*}
\int_G U(g)^{-1} K(\pi,\pi) U(g) \de\mu_G (g) = I_{\hh_\pi} \otimes A ,
\end{equation*}
with $A\in\elle{\kk_\pi}$. If $T_\pi\in\ti (\kk_\pi )$, then
\begin{eqnarray*}
\tr{AT_\pi} & = & d_\pi^{-1} \tr{(I_{\hh_\pi} \otimes A) (I_{\hh_\pi} \otimes T_\pi)} \\
& = & d_\pi^{-1} \int_G \tr{(\pi(g)^{-1} \otimes I_{\kk_\pi} ) K(\pi,\pi) (\pi(g) \otimes I_{\kk_\pi} ) (I_{\hh_\pi} \otimes T_\pi)} \de\mu_G (g) \\
\textrm{\emph{cyclicity of} ${\rm tr}$} & = & d_\pi^{-1} \tr{K(\pi,\pi) (I_{\hh_\pi} \otimes T_\pi)} = d_\pi^{-1} \tr{ \traccabis{\pi}{K(\pi,\pi)} T_\pi } ,
\end{eqnarray*}
i.e.~$A = d_\pi^{-1} \traccabis{\pi}{K(\pi,\pi)}$. Equation (\ref{traccia di K}) then follows from (\ref{normal. di Pi}).
\end{enumerate}

We denote by $\cc$ the convex set of maps $K : \hat{G} \times \hat{G} \frecc \elle{\hh}$ satisfying conditions {\rm (i')}-{\rm (iv')}. The following simple remarks will be useful later.

\begin{remark}\label{remark:k-in-c}
By eq.~(\ref{annullamento di Pi}), if $K\in\cc$, then $(\ker K(\rho , \pi) )^\perp \subseteq \hh_\pi \otimes \kk_\pi$ and $\ran K(\rho , \pi) \subseteq \hh_\rho \otimes \kk_\rho$.
\end{remark}

\begin{remark}\label{remark:k-dd}
By Corollary \ref{forma di Pi}, the kernels $K\in\cc$ are the functions $K:\hat{G}\times\hat{G} \frecc \lh$ given by
$$
\ip{w}{K(\rho,\pi)v} = \sqrt{d_\pi d_\rho} \int_H \ip{\contrbis{\rho}{ VU(h) w}}{\contrbis{\pi}{ VU(h) v}} \de \mu_H (h)
$$
for $V = \oplus_\pi I_{\hh_\pi} \otimes V_\pi$ with $V_\pi : \kk_\pi \frecc \hh_\pi^\ast \otimes \kk$ isometries. In particular
\begin{equation}\label{magg. di K}
\no{K(\rho,\pi)} \leq \sqrt{d_\pi d_\rho} \no{{\rm ctr}_{\hh_\rho}} \no{{\rm ctr}_{\hh_\pi}} = d_\pi d_\rho .
\end{equation}
\end{remark}

Proposition \ref{La proposizione col kernel} can be restated in the following form.
\begin{corollary}\label{cor. sull'isom. delle POVM con cc}
There is a convex isomorphism between $\cc$ and the set of $U$-covariant POVMs on $\Omega$. If $K\in\cc$, the corresponding POVM $\E$ is given by
\begin{equation*}
\ip{w_\rho}{\E(X) v_\pi} = \int_X \ip{U(g)^* w_\rho}{ K(\rho,\pi) U(g)^* v_\pi } \de\mu_\Omega (\dot{g})
\end{equation*}
for all $v_\pi \in\hh_\pi \otimes \kk_\pi$, $w_\rho \in\hh_\rho \otimes \kk_\rho$.
\end{corollary}

\section{Extremal covariant POVMs}\label{Extremal}

\subsection{Topology on POVMs}

The convex set of POVMs $\E : \boro \frecc \lh$ has a natural compact topology, under which $U$-covariant POVMs form a closed convex subset. Krein-Millman theorem then asserts that every $U$-covariant POVM can be approximated by a convex sum of extremal $U$-covariant POVMs. Moreover, if the convex set $\cc$ introduced in the last section is endowed with the topology of pointwise ultraweak convergence, the isomorphism estabilished in Corollary \ref{cor. sull'isom. delle POVM con cc} actually becomes a homeomorphism.

The rest of this section is devoted to the definition and properties of the topologies on covariant POVMs and on $\cc$. In the next section, we will characterize the extreme points of the set of $U$-covariant POVMs.

Let $M(\Omega)$ be the partially ordered Banach space of the complex Borel measures on $\Omega$ with the norm of the total variation. We will consider an element $\mu\in M(\Omega)$ both as a map from $\boro$ to $\C$ and as a bounded linear functional on the Banach space $C(\Omega)$ (the space of continuous functions $f:\Omega\frecc\C$ with the norm $\no{f} = \max_{x\in\Omega} |f(x)|$).

Let us denote by $\elle{\trh ; M(\Omega)}$ the Banach space of the bounded linear maps from $\trh$ into $M(\Omega)$. The space $\elle{\trh ; M(\Omega)}$ has a natural ordering arising from the definition that a map $\mm : \trh \frecc M(\Omega)$ is positive if $\mm (T) \geq 0$ for all $T\geq 0$.

\begin{definition}\label{def. di POVM 2}
An {\em operator valued measure} (OVM) based on $\Omega$ is any element $\mm\in\elle{\trh ; M(\Omega)}$. We use the abbreviated notation $M (\Omega ; \hh) := \elle{\trh ; M(\Omega)}$.

A {\em normalized positive operator valued measure} (POVM) based on $\Omega$ is any positive element $\ee\in M (\Omega ; \hh)$ such that $\ee [T] (\Omega) = \tr{T}$ (equivalently, $\ee [T] ( 1 ) = \tr{T}$) for all $T\in\trh$. We denote by $P (\Omega ; \hh)$ the set of POVMs on $\Omega$.

A {\em $U$-covariant POVM} is an element $\ee\in P (\Omega ; \hh)$ such that $\ee [U(g) T U(g)^{-1}] (X) = \ee [T] (g^{-1} X)$ for all $X\in\boro$ (equivalently, $\ee [U(g) T U(g)^{-1}] (f) = \ee [T] (f^g)$ for all $f\in C(\Omega)$, where $f^g (x) = f (gx)$) and for all $T\in\trh$, $g\in G$. We denote by $P_U (\Omega ; \hh)$ the set of $U$-covariant POVMs on $\Omega$.
\end{definition}

The above definition of POVM agrees with Definition \ref{def. di POVM 1} in the following sense. If $\mm$ is an OVM, for all $X\in \boro$ it defines an element $\M (X) \in \lh$ such that $\tr{T \M (X)} = \mm [T] (X)$ $\forall T\in\trh$. POVMs in the sense of Definition \ref{def. di POVM 2} are then all the OVMs $\mm$ such that $\M : \boro \frecc \lh$ is a POVM in the sense of Definition \ref{def. di POVM 1}. Moreover, $\mm$ is $U$-covariant if and only if $\M$ is $U$-covariant.

In $M (\Omega ; \hh)$ we define the following family of seminorms labeled by $f\in C(\Omega)$, $T\in\trh$
$$
\no{\mm}_{T;f} = | \mm [T] (f) |.
$$
Let $M_w (\Omega ; \hh)$ be the set $M (\Omega ; \hh)$ endowed with the locally convex topology generated by these seminorms. It is easy to see that $P (\Omega ; \hh)$ and $P_U (\Omega ; \hh)$ are closed convex subsets of $M_w (\Omega ; \hh)$. Moreover, $P (\Omega ; \hh)$ is norm bounded. In fact, let $\ee$ be a POVM, and, for $T\in\trh$, decompose $T = T_{11} - T_{12} + i ( T_{21} - T_{22} )$, with $T_{jk} \geq 0$ and $\notr{T_{j1}} + \notr{T_{j2}} \leq \notr{T}$. We have
$$
\no{\ee [T]} \leq \sum_{j,k} \no{\ee [T_{jk}]} = \sum_{j,k} \ee [T_{jk}] (\Omega) = \sum_{j,k} \notr{T_{jk}} \leq 2 \notr{T} ,
$$
thus showing $\no{\ee} \leq 2$, as claimed.

The following fact has a routine proof, which we omit (see, for instance, the proof of Theorem~2.5.2 in \cite{AN89}).
\begin{proposition}
The unit ball in $M_w (\Omega ; \hh)$ is compact.
\end{proposition}
By Krein-Millman theorem we then have the following corollary.
\begin{corollary}
$P_U (\Omega ; \hh)$ is the closed convex hull of its extreme points.
\end{corollary}

Let $\bb$ be the linear space of functions $F : \hat{G} \times \hat{G} \frecc \lh$ endowed with the topology of pointwise ultraweak convergence. In other words, $\bb$ has the topology generated by the family of seminorms
$$
\no{F}_{T;\rho,\pi} = \left| \tr{TF(\rho,\pi)} \right|, \quad T\in\trh, \, \rho, \pi \in \hat{G}.
$$
It is easy to check that the convex subset $\cc\subseteq \bb$ is closed.
\begin{lemma}\label{lemma sulla conv. in cc}
Suppose $\ti_0(\hh) \subseteq \trh$ is a dense subset. Then, a net $(K_\lambda)_{\lambda\in\Lambda}$ converges to $K$ in $\cc$ if and only if $\no{K_\lambda - K}_{T_0;\rho,\pi} \to 0$ for all $\rho,\pi\in \hat{G}$ and $T_0\in\ti_0 (\hh)$.
\end{lemma}
\begin{proof}
This follows easily by the inequality
\begin{eqnarray*}
\no{K_\lambda - K}_{T;\rho,\pi} & \leq & \no{T-T_0} \no{K_\lambda (\rho,\pi)} + \no{K_\lambda - K}_{T_0;\rho,\pi} +
\no{T-T_0} \no{K(\rho,\pi)} \\
& \leq & 2 d_\rho d_\pi \no{T-T_0} + \no{K_\lambda - K}_{T_0;\rho,\pi}
\end{eqnarray*}
for all $T\in\ti (\hh)$, $T_0\in\ti_0 (\hh)$, where we have used inequality (\ref{magg. di K}) derived in Remark \ref{remark:k-dd}.
\end{proof}

We denote by $j : P_U (\Omega ; \hh) \frecc \cc$ the convex set isomorphism established by Corollary \ref{cor. sull'isom. delle POVM con cc}.
\begin{proposition}
The map $j$ is a homeomorphism.
\end{proposition}
\begin{proof}
Since $P_U (\Omega ; \hh)$ is compact, it is enough to check that $j$ is continuous. Let $(\ee_\lambda)_{\lambda\in\Lambda}$ be a net converging to $\ee$ in $P_U (\Omega ; \hh)$, and let $K_\lambda = j (\ee_\lambda)$, $K = j (\ee)$. Fix $v,w\in\hh$, and let $v_\pi , w_\rho$ be the components of $v,w$ in $\hh_\pi \otimes \kk_\pi$ and $\hh_\rho \otimes \kk_\rho$, respectively. Let $f\in C(\Omega)$. Define the trace class operator $T = \int_\Omega f(\dot{g}) \kb{U(g)^{-1} v}{U(g)^{-1} w} \de\mu_\Omega (\dot{g})$. We then have
\begin{eqnarray*}
\no{K_\lambda - K}_{T;\rho,\pi} & = & \left| \int_\Omega f(\dot{g}) \ip{U(g)^{-1} w}{[K_\lambda (\rho,\pi) - K (\rho,\pi)] U(g)^{-1} v} \de\mu_\Omega (\dot{g}) \right| \\
& = & \no{\ee_\lambda - \ee}_{\kb{v_\pi}{w_\rho} \, ; f} \to 0 .
\end{eqnarray*}
Choosing as $f$ a Dirac sequence in $C(\Omega)$, we have $T \to \kb{v}{w}$ in $\trh$, and, since the linear span of rank one operators is dense in $\trh$, $K_\lambda \to K$ by Lemma \ref{lemma sulla conv. in cc}.
\end{proof}

\subsection{Characterization of extreme points}

As we have seen, the map $j : P_U (\Omega ; \hh) \frecc \cc$ estabilished in Corollary \ref{cor. sull'isom. delle POVM con cc} is an isomorphism of convex compact topological spaces. In this section, we determine the extreme points of $P_U (\Omega ; \hh)$ by characterising the extreme points of $\cc$.

Let $\ff (\hat{G} ; \hh)$ be the space of functions $f : \hat{G} \frecc \hh$.
For a fixed $K\in\cc$, we define the following linear subspace of $\ff (\hat{G} ; \hh)$
\begin{equation*}
\hh_K^0 := \spanno{ K(\cdot, \pi) v \mid \pi\in \hat{G}, \, v\in\hh}.
\end{equation*}
We set
$$
\ip{K(\cdot, \rho) w}{K(\cdot, \pi) v}_K = \ip{w}{K(\rho, \pi) v},
$$
and then $\scal{\cdot}{\cdot}_K$ extends to a scalar product on $\hh_K^0$. We denote by $\hh_K$ the completion of $\hh_K^0$ with respect to $\scal{\cdot}{\cdot}_K$. It can be shown that $\hh_K$ is still a subspace of $\ff (\hat{G} ; \hh)$ and that the evaluation maps
$$
\ev_\pi : \hh_K \frecc \hh , \qquad \ev_\pi f = f(\pi)
$$
are continuous for all $\pi\in\hat{G}$. Moreover, a straightforward calculation shows that
\begin{equation*}
K(\cdot, \pi) v = \ev_\pi^\ast v , \qquad K(\rho, \pi) = \ev_\rho \ev_\pi^\ast .
\end{equation*}
In particular, the set $\cup_{\pi\in\hat{G}} \ran \ev_\pi^\ast$ is total in $\hh_K$. 
The space $\hh_K$ is called the \emph{reproducing kernel Hilbert space} of $\hh$-valued functions associated to the \emph{reproducing kernel} $K$. For more details on the construction and properties of $\hh_K$ we refer to \cite{Pedrick57}.

For all $h\in H$ and $f\in\hh_K$, let
$$
[\tilde{U}(h) f] (\pi) : = U(h) f(\pi) \quad \forall \pi\in\hat{G} .
$$
By property {\rm (ii')} of $K$
$$
\tilde{U}(h) K(\cdot , \pi) v = K(\cdot , \pi) U(h) v \in \hh_K
$$
and
\begin{eqnarray*}
&& \ip{\tilde{U}(h) K(\cdot , \rho) w}{\tilde{U}(h) K(\cdot , \pi) v}_K = \ip{U(h) w}{K(\rho , \pi) U(h) v} \\
&& \qquad = \ip{w}{K(\rho , \pi) v} = \ip{K(\cdot , \rho) w}{K(\cdot , \pi) v}_K ,
\end{eqnarray*}
which shows by continuity that $\tilde{U}$ is a unitary representation of $H$ in $\hh_K$.
By definition, 
$$
\ev_\pi \tilde{U}(h) = U(h) \ev_\pi .
$$

\begin{remark}
The Hilbert space $\hh_K$ is unique in the following sense. Suppose $\tilde{\hh}$ is a Hilbert space carrying a unitary representation $V$ of $H$, and $\{ \gamma_\pi \}_{\pi \in\hat{G}}$ is a sequence of operators $\gamma_\pi : \hh \frecc \tilde{\hh}$ such that
\begin{enumerate}
\item $\gamma_\pi U(h) = V(h) \gamma_\pi$ for all $h\in\hh$;
\item $\cup_{\pi\in\hat{G}} \ran \gamma_\pi$ is total in $\tilde{\hh}$;
\item $K (\rho,\pi) = \gamma_\rho^\ast \gamma_\pi$ for all $\rho,\pi\in\hat{G}$.
\end{enumerate}
Then there exists a unitary map $W : \hh_K \frecc \tilde{\hh}$ intertwining $\tilde{U}$ with $V$ and such that $W \ev_\pi^\ast = \gamma_\pi$ for all $\pi\in\hat{G}$.
\end{remark}

We introduce the following closed subspace of $\ti (\hh)$
$$
\ti_U := \left\{ T\in\ti (\hh) \mid TU(g) = U(g)T \quad \forall g\in G \right\} = \spannochiuso{ I_{\hh_\pi} \otimes T_\pi \mid T_\pi \in \ti (\kk_\pi) }.
$$

The next two closed subspaces of $\ti (\hh_K)$ are essential in our investigation:
\begin{eqnarray*}
\ti_{\tilde{U}} & := & \left\{ T\in \ti (\hh_K) \mid T \tilde{U}(h) = \tilde{U}(h) T \quad \forall h\in H \right\}, \\
\tilde{\ti}_U & := & \spannochiuso{\ev_\pi^\ast T \ev_\pi \mid T \in \ti_U ,\, \pi\in\hat{G}}.
\end{eqnarray*}
Clearly, $\tilde{\ti}_U \subseteq \ti_{\tilde{U}}$.

We proceed by defining a bounded mapping $\pp: \ti (\hh_K) \frecc \ti (\hh_K)$ by formula
\begin{equation*}
\pp T = \int \tilde{U}(h) T \tilde{U}(h)^\ast \de\mu_H(h).
\end{equation*}
Then $\pp^2 = \pp$ and $\ran P = \ti_{\tilde{U}}$. The adjoint $\pp^\prime : \elle{\hh_K} \frecc \elle{\hh_K}$ satisfies $\pp^{\prime 2} = \pp^\prime$, and
$$
\ran \pp^\prime = \left\{ B\in\elle{\hh_K} \mid B \tilde{U}(h) = \tilde{U}(h) B \quad \forall h\in H \right\}.
$$
It follows that
\begin{equation}\label{duale di T Utilde}
\ti_{\tilde{U}}^\ast = \left\{ B\in\elle{\hh_K} \mid B \tilde{U}(h) = \tilde{U}(h) B \quad \forall h\in H \right\}.
\end{equation}

\begin{theorem}\label{th:notin}
$K \notin \textrm{ext}\, \cc$ if and only if there exists a nonzero operator $B\in\elle{\hh_K}$ such that
\begin{itemize}
\item[(i)] $B \tilde{U}(h) = \tilde{U}(h) B$ for all $h\in H$;
\item[(ii)] $\tr{BT} = 0$ for all $T \in \tilde{\ti}_U$.
\end{itemize}
\end{theorem}
\begin{proof}
We note that the content of Theorem \ref{th:notin} remains the same if we assume that $B = B^\ast$. This follows from the fact that the adjoint operation leaves $\tilde{\ti}_U$ invariant.
\begin{itemize}
\item[$\Longleftarrow )$] Define
$$
K_{\pm} (\rho,\pi) = \ev_\rho (I\pm \no{B}^{-1} B) \ev_\pi^\ast.
$$
Since the set $\cup_{\pi\in\hat{G}} \ran \ev_\pi^\ast$ is total in $\hh_K$, $B\neq O$ implies that $K_+ \neq K_-$.

It follows from $I\pm \no{B}^{-1} B \geq O$ that 
$$
\sum_{\pi,\rho} \ip{v^\rho}{K_\pm (\rho , \pi) v^\pi} = \ip{\sum_\rho \ev_\rho^\ast v^\rho}{(I\pm \no{B}^{-1} B) \sum_\pi \ev_\pi^\ast v^\pi}_K \geq 0
$$
for all finite sequences $ \{ v^\pi \} $ in $\hh$. Hence, $K_+$ and $K_-$ are positive definite.

The operator $K_\pm (\rho, \pi)$ commutes with $U |_H $ by the intertwining properties of $\ev_\pi,\ev_\rho$ and $B$.

Finally, if $T_\rho \in \ti (\kk_\rho)$ then
\begin{eqnarray*}
\tr{T_\rho \traccabis{\rho}{K_\pm (\pi, \pi)}} & = & \tr{(I_{\hh_\rho} \otimes T_\rho) K_\pm (\pi, \pi)} \\
& = & \tr{\ev_\pi^\ast (I_{\hh_\rho} \otimes T_\rho) \ev_\pi (I\pm \no{B}^{-1} B)} \\
& = & \tr{\ev_\pi^\ast (I_{\hh_\rho} \otimes T_\rho) \ev_\pi } = \tr{(I_{\hh_\rho} \otimes T_\rho) K (\pi, \pi) } \\
& = & \tr{T_\rho \traccabis{\rho}{K (\pi, \pi)}}
\end{eqnarray*}
i.e.
$$
\traccabis{\rho}{K_\pm (\pi, \pi)} = \traccabis{\rho}{K (\pi, \pi)} =
\delta_{\rho \pi} d_\rho I_{\kk_\rho}.
$$
We conclude that $K_\pm \in \cc$. Since $K = \half (K_+ + K_-)$, this means that $K$ is not an extreme point.

\item[$\Longrightarrow )$] Suppose $K = \half (K_+ + K_-)$, with $K_\pm \in \cc$ and $K_\pm \neq K$.

Define the following sesquilinear positive definite forms $\scal{\cdot}{\cdot}_\pm$ in $\hh_K^0$
$$
\ip{\sum_\rho \ev_\rho^\ast w^\rho}{\sum_\pi \ev_\pi^\ast v^\pi}_\pm = \sum_{\pi,\rho} \ip{w^\rho}{K_\pm(\rho , \pi) v^\pi}
$$
for all finite sequences $\{ v^\pi \}$ and $\{ w^\rho \}$ in $\hh$. The forms $\scal{\cdot}{\cdot}_\pm$ are well-defined. In fact, we have
\begin{equation*}
\sum_{\pi,\rho} \ip{v^\rho}{K_\pm(\rho , \pi) v^\pi} \leq 2 \sum_{\pi,\rho} \ip{v^\rho}{K(\rho , \pi) v^\pi}
\end{equation*}
and hence
\begin{eqnarray*}
&&  \left| \ip{\sum_\rho \ev_\rho^\ast w^\rho}{\sum_\pi \ev_\pi^\ast v^\pi}_\pm \right|  \leq \textrm{\emph{by Cauchy-Schwartz}} \\
&&\qquad \leq \left[\sum_{\pi,\pi^\prime} \ip{v^{\pi^\prime}}{K_\pm(\pi^\prime , \pi) v^\pi} \right]^{1/2}
\left[\sum_{\rho,\rho^\prime} \ip{w^{\rho^\prime}}{K_\pm(\rho^\prime , \rho) w^\rho} \right]^{1/2}  \\
&&\qquad \leq 2 \left[\sum_{\pi,\pi^\prime} \ip{v^{\pi^\prime}}{K (\pi^\prime , \pi) v^\pi} \right]^{1/2}
\left[\sum_{\rho,\rho^\prime} \ip{w^{\rho^\prime}}{K (\rho^\prime , \rho) w^\rho} \right]^{1/2}  \\
&&\qquad = 2 \no{\sum_\pi \ev_\pi^\ast v^\pi}_{K} \no{\sum_\rho \ev_\rho^\ast w^\rho}_{K}.
\end{eqnarray*}
By the above equation, we see also that $\scal{\cdot}{\cdot}_\pm$ are bounded. So, there are positive operators $B_+ , B_- \in \elle{\hh_K}$ such that $\ip{g}{f}_\pm = \ip{g}{B_\pm f}_{K}$ for all $f,g \in \hh_K$. By definition, 
$K_\pm ( \pi , \pi ) = \ev_\pi B_\pm \ev_\pi^\ast$.

Set $B = B_+ - B_-$. Then, $B\neq 0$ since $K_+ \neq K_-$. Since
\begin{eqnarray*}
&&\ip{\tilde{U}(h) \ev_\rho^\ast w}{\tilde{U}(h) \ev_\pi^\ast v}_\pm = \ip{\ev_\rho^\ast U(h) w}{\ev_\pi^\ast U(h) v}_\pm \\
&&\qquad = \ip{U(h) w}{K_\pm (\rho , \pi) U(h) v}
=  \ip{w}{K_\pm (\rho , \pi) v} \\
&&\qquad = \ip{\ev_\rho^\ast w}{\ev_\pi^\ast v}_\pm
\end{eqnarray*}
and $\cup_{\pi\in\hat{G}} \ran \ev_\pi^\ast$ is total in $\hh_K$, it follows that
$B$ commutes with $\tilde{U}$.

Finally, if $T = I_{\hh_\rho} \otimes T_\rho$ with $T_\rho \in \ti(\kk_\rho)$ we have
\begin{eqnarray*}
\tr{B_\pm \ev_\pi^\ast T \ev_\pi} & = & \tr{\ev_\pi B_\pm \ev_\pi^\ast T} = \tr{K_\pm ( \pi , \pi ) T } \\
& = & \tr{T_\rho \traccabis{\rho}{K_\pm (\pi , \pi)}} = \delta_{\rho \pi} d_\rho \tr{T_\rho} .
\end{eqnarray*}
It follows that $\tr{B \ev_\pi^\ast T \ev_\pi} = 0$, i.e., $B$ satisfies item (ii) of the theorem.
\end{itemize}
\end{proof}

\begin{corollary}\label{cor:ext}
$K\in \textrm{ext}\, \cc$ if and only if $\tilde{\ti}_U = \ti_{\tilde{U}}$.
\end{corollary}
\begin{proof}
By Theorem \ref{th:notin} and equation \eqref{duale di T Utilde}, the inclusion $\tilde{\ti}_U \subseteq \ti_{\tilde{U}}$ is an equality if and only if $K\in \textrm{ext}\, \cc$.
\end{proof}

\begin{example}
Let us consider the case when $G$ is abelian, $H = \{ e \}$, and $\dim \kk_\pi = 1$ for all $\pi$. 
We fix an orthonormal basis $\{ e_\pi \}_{\pi\in\hat{G}}$ in $\hh= \oplus_{\pi\in \hat{G}} \hh_\pi$ such that $U(g) e_\pi = \pi (g) e_\pi$ for all $g\in G$. Then
$$
\ti_U = \spannochiuso{ \kb{e_\pi}{e_\pi} \mid \pi\in\hat{G}}.
$$
If $K\in\cc$, then $K(\rho , \pi) = \tilde{K} (\rho , \pi) \kb{e_\rho}{e_\pi}$, $\tilde{K}$ being a positive semidefinite complex matrix with $\tilde{K} (\pi,\pi) = 1$. This sets up an isomorphism $K\mapsto \tilde{K}$ from the convex set $\cc$ into the convex set $\tilde{\cc}$ of positive semidefinite matrices with $1$ as every diagonal element.

Let $\hh_{\tilde{K}}$ be the reproducing kernel Hilbert space of $\C$-valued functions associated to $\tilde{K}$, i.e.~there exists a total sequence $\{ \eta_\pi \}_{\pi\in\hat{G}}$ in $\hh_{\tilde{K}}$ such that $\tilde{K} (\rho,\pi) = \ip{\eta_\rho}{\eta_\pi}_{\hh_{\tilde{K}}}$.
Then we can set up a unitary isomorphism $j : \hh_K \frecc \hh_{\tilde{K}}$, given by
$$
\ip{\eta_\pi}{ j(f) }_{\hh_{\tilde{K}}} := \ip{e_\pi}{\ev_\pi f} .
$$
We thus have, with easy calculations,
$$
j(\ev_\pi^\ast e_\pi) = \eta_\pi ,
$$
and so
\begin{eqnarray*}
j \ti_{\tilde{U}} j^\ast & = & j \ti (\hh_K) j^\ast = \ti (\hh_{\tilde{K}}) \\
j \tilde{\ti}_U j^\ast & = & \spannochiuso{ \kb{j \ev_\pi^\ast e_\pi}{ j \ev_\pi^\ast e_\pi} \mid \pi \in\hat{G}} = \spannochiuso{ \kb{\eta_\pi}{\eta_\pi} \mid \pi \in\hat{G}} .
\end{eqnarray*}
Corollary \ref{cor:ext} then says that $\tilde{K} \in \textrm{ext}\, \tilde{\cc}$ if and only if
$$
\spannochiuso{ \kb{\eta_\pi}{ \eta_\pi} \mid \pi \in\hat{G}} = \ti (\hh_{\tilde{K}}).
$$
This result has also been derived in \cite{KiPe06}.
\end{example}

\subsection{Rank 1 extremals}\label{Rank}

Suppose $K\in\cc$. Let ${\rm rank}\, K$ denote the dimension of $\hh_K$ (or, equivalently, the dimension of $\hh_K^0$). Next we assume that ${\rm rank}\, K=1$. This means that there exists $f\in\hh_K$, $\langle f|f\rangle_K=1$ such that for all $\pi\in \hat{G}, \, v\in\hh$,
$K(\cdot, \pi) v=c_{\pi,v} f$ for some $c_{\pi,v}\in\C$; or more precisely
$K(\rho,\pi) v=c_{\pi,v} f(\rho).$ Thus, 
\begin{equation*}
c_{\pi,v}\langle w|{f(\rho)}\rangle = \langle w|{K(\rho,\pi) v}\rangle = \langle{K(\cdot, \rho) w}|{K(\cdot, \pi) v}\rangle_K = \overline{c_{\rho,w}}c_{\pi,v}\langle{f}|{f}\rangle_K=\overline{c_{\rho,w}} c_{\pi,v} ,
\end{equation*}
which implies that $K(\rho, \pi) v=\langle{f(\pi)}|v\rangle f(\rho)$ or
$
K(\rho,\pi)=|f(\rho)\rangle\langle{f(\pi)}|,
$
in short.

It follows from Remark \ref{remark:k-in-c} that $f(\rho)\in\hh_\rho\otimes\kk_\rho$, so that we can write
$$
f(\rho)=\sum_{m,n}f^\rho_{mn}h^\rho_m\otimes k^\rho_n
$$
where $\{h^\rho_m\}_{m=1}^{d_\rho}$ and $\{k^\rho_n\}_{n=1}^{\dim\kk_\rho}$
are orthonormal bases of $\hh_\rho$ and $\kk_\rho$, respectively.
Since $\traccabis{\rho}{K(\pi,\pi)} = \delta_{\rho \pi} d_\rho I_{\kk_\rho}$
and
$$
\traccabis{\rho}{|f(\rho)\rangle\langle{f(\rho)}|}=\sum_{n,l}
\left(\sum_m f_{mn}^\rho\overline{f_{ml}^\rho}\right)
|k_n^\rho\rangle\langle k_l^\rho|
$$
one gets
$$
\sum_m f_{mn}^\rho\overline{f_{ml}^\rho}=d_\rho\delta_{nl}.
$$
By defining 
$$
f^\rho_n:=d_\rho^{-1/2}\sum_{m}f^\rho_{mn}h^\rho_m
$$
we see that the above condition equals $\langle f^\rho_l|f^\rho_n\rangle_{\hh_\rho}=\delta_{ln}$, so that
$\{f^\rho_n\}_{n=1}^{\dim \kk_\rho}$ is an orthonormal set of $\hh_\rho$.
Hence, necessarily
$$
\dim\kk_\rho\le d_\rho
$$
and
$$
f(\rho)=d_\rho^{1/2}\sum_{n}f^\rho_{n}\otimes k^\rho_n.
$$

Since $\dim \hh_K = 1$, the representation $\tilde{U}$ reduces to a character $\lambda$ of $H$. This means $U(h) f(\rho) = [\tilde{U}(h) f](\rho) = \lambda(h) f (\rho)$, i.e., $\sum_{n} \rho(h) f^\rho_{n} \otimes k^\rho_n = \lambda(h) \sum_{n} f^\rho_{n} \otimes k^\rho_n$. Hence,
$$
\rho(h) f^\rho_{n} = \lambda(h) f^\rho_{n} \quad \forall h\in H, \, n=1,2,\ldots \dim \kk_\rho .
$$

We have arrived at the following result.
\begin{proposition}\label{rank1}
There exists rank $1$ kernels if and only if
\begin{enumerate}
\item $\dim\kk_\rho\le d_\rho$ for all $\rho\in\hat G$;
\item there exists a character $\lambda\in\hat{H}$ and, for all $\rho\in\hat G$, a subspace $\hh_\rho^\prime \subseteq \hh_\rho$ such that
\begin{enumerate}
\item $\dim \hh_\rho^\prime \geq \dim \kk_\rho$,
\item $\rho (h) |_{\hh_\rho^\prime} = \lambda (h)$ for all $h\in H$.
\end{enumerate}
\end{enumerate}
In this case, fix an orthonormal basis $\{ k_n^\sigma \}_{n=1}^{\dim\kk_\sigma}$ of $\kk_\sigma$ for all $\sigma\in\hat{G}$. Then a kernel $K$ is a rank $1$ element of $\cc$ if and only if
$$
K(\rho,\pi)=\sqrt{d_\rho d_\pi}\sum_{n=1}^{\dim\kk_\rho}\sum_{m=1}^{\dim\kk_\pi}
|f_n^\rho\otimes k_n^\rho\rangle\langle f_m^\pi\otimes k_m^\pi|
$$
where $\{f^\sigma_n\}_{n=1}^{\dim\kk_\sigma}$ is any orthonormal set in $\hh^\prime_\sigma$, $\sigma\in\hat G$.
\end{proposition}

\begin{proof}
We still need to prove the converse of the above proposition. With 
$$
f(\rho) := d_\rho^{1/2} \sum_{n=1}^{\dim\kk_\rho}f_n^\rho\otimes k_n^\rho,
$$
we see that $K(\rho,\pi) = \kb{f(\rho)}{f(\pi)}$, hence $K$ is of positive type. By item (2), we clearly have $U(h) K(\rho,\pi) = K(\rho,\pi) U(h)$ for all $h\in H$. One easily checks $\traccabis{\rho}{K(\pi,\pi)} = d_\rho \delta_{\rho \pi} I_{\kk_\rho}$. Finally, $\hh_K = \C f$, so ${\rm rank}\, K = 1$.
\end{proof}
In the finite dimensional case, it was already noted in \cite[Proposition~1]{ChDa04} that there do not exist rank 1 kernels if $\dim \kk_\rho > d_\rho$ for some $\rho\in \hat{G}$.

If $\dim\hh < \infty$, the next proposition is \cite[Corollary~1]{ChDa04}.
\begin{proposition}
Any rank 1 kernel $K$ is extremal.
\end{proposition}

\begin{proof}
Since the representation $\tilde{U}$ reduces to the character $\lambda$, it follows that $\ti_{\tilde{U}} = \ti (\hh_K)$, and so $\dim\ti_{\tilde{U}} = \dim \ti (\hh_K) = 1$.
Moreover, $K$ cannot be the zero kernel, so that $\dim\kk_\pi\ne 0$ for some $\pi\in\hat G$. For such $\pi$, let $T_\pi \in \ti (\kk_\pi)$ with $\tr{T_\pi} = 1$. We have $I_{\hh_\pi} \otimes T_\pi \in \ti_U$, and
$$
\ip{f}{\ev_\pi^\ast (I_{\hh_\pi} \otimes T_\pi) \ev_\pi f}_K = \ip{f(\pi)}{(I_{\hh_\pi} \otimes T_\pi) f(\pi)} = d_\pi \tr{T_\pi} \ne 0.
$$
Therefore, $\dim \tilde{\ti}_U = 1$,
and the result follows from Corollary \ref{cor:ext}.
\end{proof}


\begin{thebibliography}{10}

\bibitem{PSAQT82}
A.~S.~Holevo.
\newblock {\em Probabilistic and Statistical Aspects of Quantum Theory}.
\newblock North-Holland Publishing Co., Amsterdam, 1982.

\bibitem{OQP97}
P.~Busch, M.~Grabowski, and P.~J.~Lahti.
\newblock {\em Operational Quantum Physics}.
\newblock Springer-Verlag, Berlin, 1997.
\newblock second corrected printing.

\bibitem{Holevo79}
A.~S.~Holevo.
\newblock Covariant measurements and uncertainty relations.
\newblock {\em Rep. Math. Phys.}, 16:385--400, 1979.

\bibitem{LiTa94}
C.~K.~Li and B.~S.~Tam.
\newblock A note on extreme correlation matrices.
\newblock {\em SIAM J. Matrix Anal. Appl.}, 15:903--908, 1994.

\bibitem{KiPe06}
J.~Kiukas and J.~P.~Pellonp{\"a}{\"a}.
\newblock A note on infinite extreme correlation matrices.
\newblock {\em Linear Algebra Appl.}, in press.
\newblock math-ph/0612537.

\bibitem{Dariano04}
G.~M.~D'Ariano.
\newblock Extremal covariant quantum operations and positive operator valued
  measures.
\newblock {\em J. Math. Phys.}, 45:3620--3635, 2004.

\bibitem{ChDa04}
G.~Chiribella and G.~M.~D'Ariano.
\newblock Extremal covariant positive operator valued measures.
\newblock {\em J. Math. Phys.}, 45:4435--4447, 2004. 

\bibitem{Kholevo87}
A.~S.~Kholevo.
\newblock On a generalization of canonical quantization.
\newblock {\em Math. USSR Izvestiya}, 28:175--188, 1987.

\bibitem{DS}
N.~Dunford and J.~T.~Schwartz
\newblock {\em Linear Operators. I. General Theory}.
\newblock Interscience Publishers, Inc., New York, 1958.

\bibitem{CAHA95}
G.~B.~Folland.
\newblock {\em A Course in Abstract Harmonic Analysis}.
\newblock CRC Press, Boca Raton, FL, 1995.

\bibitem{QTOS76}
E.~B.~Davies.
\newblock {\em Quantum Theory of Open Systems}.
\newblock Academic Press, London, 1976.

\bibitem{Mackey52}
G.~W.~Mackey.
\newblock Induced representations of locally compact groups. {I}.
\newblock {\em Ann. of Math. (2)}, 55:101--139, 1952.

\bibitem{Cattaneo79}
U.~Cattaneo.
\newblock On {M}ackey's imprimitivity theorem.
\newblock {\em Comment. Math. Helv.}, 54:629--641, 1979.

\bibitem{CaDeTo04}
G.~Cassinelli, E.~De~Vito, and A.~Toigo.
\newblock Positive operator valued measures covariant with respect to an
  abelian group.
\newblock {\em J. Math. Phys.}, 45:418--433, 2004.

\bibitem{AN89}
G.~Pedersen.
\newblock {\em Analysis now}.
\newblock Springer-Verlag, New York, 1989.

\bibitem{Pedrick57}
G.~Pedrick.
\newblock {\em Theory of Reproducing Kernels of Hilbert Spaces of Vector Valued
  Functions}.
\newblock University of Kansas, Department of Mathematics, Lawrence, Kansas,
  1957.

\end{thebibliography}
\end{document}